\documentclass[journal,twoside,web]{ieeecolor}
\usepackage{generic}
\usepackage{cite} 
\usepackage{algorithm}
\usepackage{algpseudocode}
\usepackage{graphicx}
\usepackage{textcomp}

\usepackage{graphicx}      % include this line if your document contains figures

\usepackage{amsmath,amssymb,amsfonts,xpatch,dsfont}
\newcommand{\mat}[1]{\left[\; \begin{matrix} #1 \end{matrix} \:\right]}
\newcommand{\simode}[1]{\left\{\;
\begin{aligned} #1 \end{aligned} \right.}
\newcommand{\bm}[1]{\boldsymbol{#1}}
\newcommand{\tr}{{\sf T}}

\newtheorem{theorem}{Theorem}

\newtheorem{remark}{Remark}

\def\BibTeX{{\rm B\kern-.05em{\sc i\kern-.025em b}\kern-.08em
    T\kern-.1667em\lower.7ex\hbox{E}\kern-.125emX}}
\begin{document}
\title{Relations Between Generalized JST Algorithm and Kalman Filtering Algorithm for Time Scale Generation}
\author{Yuyue Yan, \IEEEmembership{Member, IEEE}, Takahiro Kawaguchi, \IEEEmembership{Member, IEEE}, Yuichiro Yano, \IEEEmembership{Member, IEEE}, \\ 
Yuko Hanado,   and Takayuki Ishizaki,  \IEEEmembership{Member, IEEE} 
\thanks{This work is supported by the Ministry of Internal Affairs and Communications (MIC) under its "Research and Development for Expansion of Radio Resources (JPJ000254)" program. A preliminary version \cite{yanequiva23} of the results of this paper is accepted and will be presented at 
ISPCS 2023, London, England.} 
\thanks{Yuyue Yan and Takayuki Ishizaki are with the Department of Systems and Control Engineering, Tokyo Institute of Technology, Meguro, Tokyo 152-8552 Japan (e-mail:  yan.y.ac@m.titech.ac.jp,  ishizaki@sc.e.titech.ac.jp).
  Takahiro Kawaguchi is with the Division of Electronics and Informatics, Gunma University, Kiryu,
Gunma 371-8510 Japan (e-mail:  kawaguchi@gunma-u.ac.jp).  Yuichiro Yano and Yuko Hanado are with the National Institute of Information and Communications Technology,
Koganei, Tokyo 184-0015 Japan (e-mail: y-yano@nict.go.jp, yuko@nict.go.jp).   
}
\vspace{-10pt}}
\maketitle

\begin{abstract}
In this paper, we present a generalized Japan Standard Time algorithm (JST-algo) for higher-order atomic clock ensembles and mathematically clarify the relations of the (generalized) JST-algo and the conventional Kalman filtering algorithm (CKF-algo) in the averaged atomic time and the clock residuals for   time scale generation. 
In particular, we reveal the fact that the averaged atomic time of the generalized JST-algo does not depend on the observation noise even though the measurement signal is not filtered in the algorithm. 
Furthermore, the prediction error of CKF-algo is rigorously shown by using the prediction error regarding an observable state space.  
It is mathematically shown that when the covariance matrices of system noises are identical for all atomic clocks,  considering equal averaging weights for the clocks is a necessary and sufficient condition to ensure equivalence between the generalized JST-algo and CKF-algo in averaged atomic time.
In such homogeneous systems, a necessary and sufficient condition for observation noises is presented to determine which algorithm can generate the clock residuals with smaller variances. 
A couple of numerical examples comparing the generalized JST-algo and CKF-algo are provided to illustrate the efficacy of the results. % and to reveal that if the observation noises are tiny enough, the generalized JST-algo may be replaceable from Kalman filtering algorithm for time scale generation.
\end{abstract}

\vspace{-10pt}
\section{Introduction}
\label{sec:introduction}
\IEEEPARstart{A}{n} 
atomic clock ensemble is a collection of highly accurate atomic clocks that work together to achieve a precise and stable timekeeping system. Atomic clocks are devices that measure time based on the vibrations of atoms with constant resonance frequencies, e.g., cesium and rubidium atoms. 
Even though individual atomic clocks can be accurate, they still have some tiny variations due to the environmental factors like temperature fluctuations, external electromagnetic fields, and quantum mechanical effects.
By combining the measurements from multiple atomic clocks within an ensemble, the national metrology institutes (NMIs) all over the world can reduce these variations and create a more reliable and robust timekeeping system   \cite{galleani2010time,chan2009self,liu2021improving}.
The advancements in atomic clock technology and the development of accurate time scales based on these ensembles offer numerous benefits and applications that are crucial for the future smart society, e.g., 
satellite navigation \cite{wu2015uncertainty},  
financial networks with high-frequency trading and  
time-sensitive transactions  \cite{mulvin2017media},
telecommunications \cite{seidelmann2011time},  etc. 

The variations in tick rates of atomic clocks are referred to as time deviations from the ideal clock behavior, which can be modeled as stochastic processes.
To improve the accuracy of time scales, the researchers obtained experimental evidence for modeling the behavior of atomic clocks and the time deviations the clocks experience as a series of linear stochastic differential equations \cite{zucca2005clock}.
Based on this, in the task of time generation, how to properly deal with the prediction problem for time deviations is the main issue to guarantee excellent performance of the generated time \cite{galleani2010time}. 
The algorithm that deals with such a prediction problem of time deviations is referred to as the algorithm of averaged atomic time.

The Kalman filter is a mathematical method used for state estimation/prediction in control theory and signal processing  to estimate the state of a dynamic system based on a series of noisy measurements.
It is well known that the key advantage of the Kalman filter is its ability to provide an optimal estimate by dynamically balancing the trade-off between the model predictions and the actual measurements, effectively reducing the impact of observation noises \cite{valenti2015linear,song2021robust,baraldi2012kalman}. 
In time and frequency community, the Kalman filter has been constructed as the algorithm of averaged atomic time using the difference of clock reading between two clocks as measurement signals \cite{galleani2003use,galleani2010time,greenhall2012review,trainotti2022detection,mostafa2020enhancing}. 
However, 
it is reported that one may face the numerical instability problem of the Kalman filter in the atomic clock ensembles \cite{brown1991theory,galleani2010time,yan2023structured,yandecom23}.
This is because practical implementations often ignore the fact that the dynamic system of atomic clock ensembles is undetectable\footnote{One does not have sufficient information to determine the full state of the system, and the unobservable portion of the state is not stable.}, 
whereas detectability is a necessary condition ensuring asymptotic convergence of error covariances 
\cite{de1986riccati,auger2013industrial}.
In atomic clock ensembles, since the detectability condition is broken, 
the computational errors in error covariances of the Kalman filter grow unboundedly and hence lead to numerical instability problem \cite{greenhall2012review}.  

Except for the CKF-algo,   
the JST-algo \cite{hanado2006improvement}, is specific to Japan's timekeeping system governed by the National Institute of Information and Communications Technology (NICT), 
while Japan is known for its advanced technology and precision timekeeping capabilities.
The detailed algorithm can be found in  \cite{hanado2006improvement,banerjee2008time,hanado2011overview} (and the reference therein) and is applicable for the atomic clock ensembles with second-order clocks. 
But there is no literature discussing the relation between the JST-algo and the CKF-algo. 
It is interesting to compare the two methods and ask when they are equal and which method is superior to the other.
  
In this paper, we mathematically clarify the relation among the algorithms of averaged atomic time by the CKF-algo and JST-algo. 
Different from the preliminary version \cite{yanequiva23}, we generalize the results for higher-order atomic clock ensembles.
Specifically, we generalize the existing JST-algo to a generalized form for the atomic clock ensembles with higher-order models, where the proposed generalized JST-algo is reduced to the existing JST-algo for second-order clocks. 
By theoretically analyzing the generalized JST-algo in state-space model, we reveal the fact that the averaged atomic time of JST-algo does not depend on the observation noise even though the measurement signal is not filtered in the algorithm. 
Furthermore, using observable Kalman canonical decomposition, the prediction error of CKF-algo is derived.   
We show the fact that when the covariance matrix of the system noises is identical for all atomic clocks, considering equal averaging weights for the clocks is a sufficient and necessary condition guaranteeing equivalence between the generalized JST-algo and CKF-algo in averaged atomic time.
In addition, different from the preliminary version \cite{yanequiva23}, we further clarify the relation between JST-algo and CKF-algo in individual clock residuals for the homogeneous clock ensemble, and present the sufficient and necessary conditions for observation noises to determine which algorithm can generate the clock residuals with smaller variances.

\noindent \textbf{Notation}~We write $\mathbb{R}$ for the set of real numbers, $\mathbb{R}_+$ for the set of positive real numbers, 
$\mathbb{R}^{n}$ for the set of \emph{n}$\times 1$  real column vectors,  and $\mathbb{R}$$^{n \times m}$ for the set of \emph{n}$\times m$ real matrices.
Moreover, $\otimes$ denotes the Kronecker product, 
$(\cdot)^\tr$ denotes transpose,   
$(\cdot)^\dagger$ denotes the Moore-Penrose pseudoinverse,   and
$ \mathsf{diag}(\cdot)$ denotes a diagonal matrix. 
Furthermore, $ \mathbb{E}[x]$  and $\mathbb{C}[x]$  denotes the mean value and covariance of a random variable $x$,  
Finally,  $\mathds{1}_n $ and $I_{n}$  denote the all-ones column vector
and the identity matrix of dimension $n$, respectively, and
\[
e_{1:n-1}:=\mat{I_{n-1}\\0},\quad K^0:=\mat{0 & 0 \\ 0 & K}.
\]

\section{Preliminaries}\label{sec:pre} 
\subsection{Atomic Clock}\vspace{-2pt}
Consider an atomic clock ensemble composed of $m$ clocks.
Each clock works as an independent oscillator that generates a sinusoidal signal.
The number of waves is counted as its clock reading which is slightly different from the ideal clock reading where the difference is referred to as the time deviation  (or, equivalently, so-called phase deviation).
Depending on the material property of the atomic clocks,   
the time deviation  of clock $j$ from the ideal clock is known to satisfy the $n$th order stochastic differential equation given by 
 \begin{equation} \label{eq:intmodel} 
     \Delta h^j(t) \! =\!  \sum_{i=1}^n \frac{\alpha_{i}^jt^{i-1} }{(i-1)!}
    + \sum_{i=1}^n \! \int_0^t   \!  \! \int_0^{t_1}  \! \!  \!   \cdots  \!  \!  \int_0^{t_{i-1}} \!   \! \xi_i^j(t_i) 
    dt_i \cdots dt_2 dt_1,
\end{equation}
where $\alpha_{i}^j\in\mathbb R$, $i=1,\ldots,n$, are the parameters with respect to the initial state of clock $j$, and  $ \xi_i^j\in\mathbb R$, $i=1,\ldots,n$, denote $n$ independent one-dimensional Gaussian random noises with the variance given by $\sigma_i^j\geq0$, $i=1,\ldots,n$.
For example, it is often assumed that the order $n$ is given by $n=2$ for Cesium clocks  \cite{zucca2005clock}.

\subsection{Task of Time Generation}\vspace{-2pt}
Consider a discrete-time sequence $\{t_{kT}\}_{k=0,1,2,\ldots}$ with an operating period $T\in\mathbb R_+$.
In practice, it is usually assumed that a reference signal such as UTC$(t)$ (Coordinated Universal Time) is available as an external source for the local time scale generation only at $t=kT$, 
where UTC$(t)$ can be regarded as the (approximated) ideal time.
Therefore, the time deviations $\{\Delta h^j(t)\}$ are available at $t=kT$.
However, during the time interval $t\in(t_0,t_T)$, the time deviations of the atomic clocks are never measurable because the ideal time is not available. 

Even though the time deviations of  clocks are immeasurable, 
the difference $y_{ij}=\Delta h^i(t)-\Delta h^j(t)=h^i(t)-h^j(t)$ between  clocks $i$ and $j$ is the measurable signal in the clock ensemble, where
$h^i(t)$ represents the actual clock reading of clock $i$.
The key to generating a time scale is predicting the immeasurable time deviations of the atomic clocks using the measurements during the time interval $t\in(t_0,t_T)$.
The generated time is given as
\begin{align}\label{eq:gt}  
\hat h_0(t)=\sum_{i=1}^m \beta_i\left[h^i(t)-\Delta\hat h^i(t)\right]\in\mathbb R,\quad t\in(t_0,t_T),
\end{align}  
where $\beta:=(\beta_1,\cdots,\beta_m)^{\tr}$ denotes the weights chosen based on the reliability of the individual clocks with $\beta_1+\ldots+\beta_m=1$.
Here, since the ideal clock reading $h_0(t)$ can be expressed as $h_0(t)=\hat h_0(t)$ with $\Delta\hat h^i(t)$ replaced by $\Delta h^i(t)$,
the accuracy of the generated time \eqref{eq:gt}   is evaluated by averaged \emph{atomic time} (so-called the ensemble time scale in \cite{galleani2010time})
\begin{align}\label{eq:TA} 
{\rm TA}(t) := \hat h_0(t) - h_0(t)
=\sum_{i=1}^m \beta_i
[\Delta h^i(t)-\Delta \hat h^i(t)],
\end{align}
which is the weighted  prediction error of time deviations. % on the adopted prediction algorithm.
The structure of the clock ensemble is summarized in Fig.~\ref{structure} below.

\begin{figure}
\centering
\includegraphics[width = .96\linewidth]{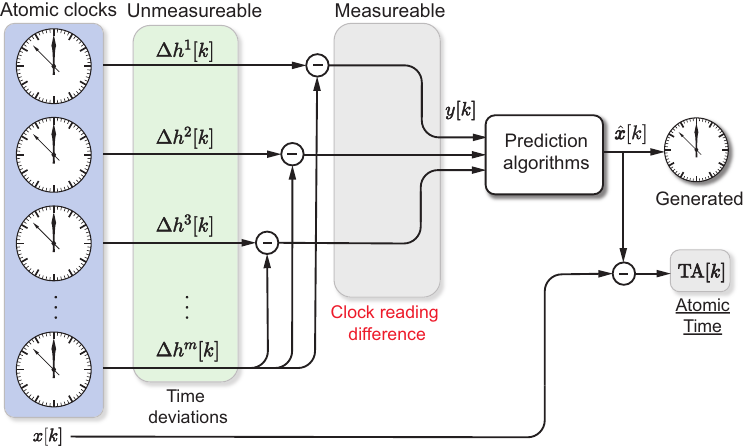}
\caption{Structure of time scale generation for an $m$-clock ensemble.}
\label{structure}
\end{figure}
\vspace{-2pt}
\subsection{State-Space Model of Clock Ensemble } 
%Before we introduce the existing algorithms of averaged atomic time, we introduce the state-space model. 
Without loss of generality, adopting clock $m$ as the reference clock in measurements, the $n$th order model \eqref{eq:intmodel} of the $m$-clock ensemble in the discrete-time sequence $\{t_k\}_{k=0,1,\ldots,T}$ during the operating interval $t\in(t_0,t_T)$ with sampling period $\tau_k=t_{k+1}-t_k\in\mathbb R_+$, $k= 0,1,\ldots,T$, is equivalent to
\begin{equation}\label{eq:Ndmodel}
\Sigma:    \simode{
    \bm{x}[k+1] &=  \bm F[k]\bm{x}[k] + \bm{v}[k]  \\
    \bm{y}[k] &= \bm H \bm{x}[k] + \bm{w}[k] \\
    \bm{x}_{\rm ens}[k] &=(I_n\otimes \beta^{\tr})\bm{x}[k]
    }
\end{equation}   
where 
$\bm{x}[k]:= (\bm{x}_1^\tr[k],  \ldots, \bm{x}_n^\tr[k])^\tr\in\mathbb R^{nm}$ is the ensemble state with
$\bm{x}_i[k]:=(x^1_i[k] ,\ldots, x^m_i[k])^\tr\in\mathbb R^{m}$, $x_{i}^j[0] = \alpha_{i}^j$ for $i=1,\ldots, n$, $j= 1,\ldots, m$;
$\bm{y}[k]= (y_{1m}[k],  \ldots, y_{(m-1)m}[k])$  is the measurement of the ensemble  including the observation noise $\bm{w}[k]\in\mathbb R^{m-1}$.
In particular, the system   matrix and the observation matrix 
\begin{align} 
\bm F[k]:= A[k]\otimes I_m, \quad
\bm H:= C \otimes \overline{V}
\end{align}  
of \eqref{eq:Ndmodel} are  defined as
\begin{align} 
A[k]:=&A(\tau_k):=\mat{
    1 & \tau_k & \frac{\tau_k^2}{2} & \cdots & \tfrac{\tau_k^{n-1}}{(n-1)!} \\
    0 & 1 & \tau_k & \cdots & \frac{\tau_k^{n-2}}{(n-2)!} \\
    \vdots & & \ddots & \ddots & \vdots \\
    \vdots & & & 1 & \tau_k \\
    0 & 0 & \cdots & \cdots & 1
    }  
\\ 
C :=& \mat{1 & 0 & \cdots & 0}\in\mathbb R^{1\times n}
\\ 
\overline{V}:=&\mat{I_{m-1}& -\mathds{1}_{m-1}}\in\mathbb R^{(m-1)\times m}.
\end{align}  
The signal $\bm{v}[k]=(\bm{v}_1^\tr[k],  \ldots, \bm{v}_n^\tr[k])\in\mathbb R^{nm}$ 
represents the system noise with 
$\bm{v}_i[k]:=(\bm v^1_i[k] ,\ldots, \bm v^m_i[k])^\tr\in\mathbb R^{m}$,
where $\bm{v}^j[k]:=(\bm v^j_1[k],\ldots ,\bm v^j_n[k])^\tr\in\mathbb R^{n}$ is the
Gaussian noise, that comes from the individual noise 
$\xi_1^j,\ldots$, 
$\xi_n^j$,  defined as 
\begin{equation}
\bm{v}^j[k] :=  \int_{0}^{\tau_k} A(t_{k+1}-t)[\xi^j_1(t),\ldots, \xi^j_n(t)]^{\tr}dt.
\end{equation}
In this state space model, the state 
${\bm{x}}_1=(\Delta h^{1},\ldots,\Delta h^{m})^\tr\in\mathbb R^{m}$ 
represents the vector of the time deviation of the clocks 
and $\bm{x}_{\rm ens}[k]\in\mathbb R^n$ represents the (weighted) ensemble state 
so that $C\bm{x}_{\rm ens}[k]\in\mathbb R$ (or equivalently, $\beta^{\tr}{\bm{x}}_1[k]\in\mathbb R$) denotes the ensemble time deviation. 

Thus, the atomic time ${\rm TA}(t)$ at $t=t_k$ can be expressed by
\begin{equation}\label{eq:TA2} 
{\rm TA}[k]=C\bm{x}_{\rm ens}[k]-C\hat{\bm{x}}_{\rm ens}[k]
= C\bm{\epsilon}_{\rm ens}[k]
\end{equation}
where
$\bm{\epsilon}_{\rm ens}[\cdot]=(I_n \otimes \beta^{\tr})\bm{\epsilon}[\cdot]\in\mathbb R^n$ is the ensemble prediction error from the prediction error $\bm{\epsilon}[\cdot]:=\bm{x}[\cdot] - \hat{\bm{x}}[\cdot]\in\mathbb R^{nm}$.
\vspace{-5pt}

\subsection{Algorithm of Averaged Atomic Time for Japan Standard Time (JST-algo)}
Similar  to the other standard time in the world,
JST is generated by integrating about 20 high-precision atomic clocks, including hydrogen-maser clocks, cesium-beam atomic clocks, and optical lattice clocks, where the atomic clocks are assumed to be in second order.
The pseudocode of JST-algo is shown in Algorithm~\ref{Hanado's loop} above, which is only applicable for the ensemble with the second-order model \eqref{eq:intmodel}  of the clocks, i.e, $n=2$.
The principle of JST-algo is explained in the following.

\begin{algorithm}[t]
\caption{JST-algo \cite{hanado2006improvement}}
\begin{algorithmic}[1]
\State {\bf Initialization}: $\Delta\hat{h}^i[0]\approx\Delta{h}^i[0]$, $\hat \alpha_2^i\approx\alpha_{2}^i$, and $k=1$
\While{$k\leq T$}
\For{$i=1,\ldots,m$} \Comment{Prediction}
%\State $\hat y_i[k]=(\Delta\hat{h}^i[k-1]-\Delta\hat{h}^i[k-2])/\tau_{k-2}$
\State $\Delta\hat{h}^i[k]=\Delta\hat{h}^i[k-1]+\hat \alpha_2^i\tau_{k-1}$
\EndFor
\State $\Delta\hat{h}^m[k]=\sum\nolimits_{i=1}^m \beta_i\left(\Delta\hat{h}^i[k]-y_{im}[k]\right)$ \Comment{Weighting}
\For{$i\not=s$}
\State $\Delta\hat{h}^i[k]=\Delta\hat{h}^m[k]+y_{im}[k]$ \Comment{Update}
\EndFor
\State $k\leftarrow k+1$
\EndWhile
\end{algorithmic}\label{Hanado's loop}
\end{algorithm}

\subsubsection{Prediction procedure}
Consider $n=2$, ignoring the terms of Gaussian random noises in \eqref{eq:intmodel}, 
we have
\begin{equation}\label{eq:motivation0}
    \Delta h^{j}(t) \approx \alpha_{1}^j + \alpha_{2}^j t.
\end{equation} 
It turns out that
\begin{equation}\label{eq:motivation1}
    \Delta h^{j}(t_{k+1}) \approx \Delta h^{j}(t_{k}) + \alpha_{2}^j \tau_k.
\end{equation} 
where the initial conditions $\alpha_{i}^j$, $i=1,2$, of the atomic clocks can be estimated using some identification methods based on the external reference UTC$(t)$, $t=kT$, and the operating interval $T$. 
For example, we can take 
\begin{align}
\alpha_{1}^j  &\approx \hat \alpha_1^j = \Delta h^{j}(t_0)
\\ \label{eq:motivation2}
\alpha_{2}^j  &\approx \hat \alpha_2^j = \frac{\Delta h^{j}(t_0) - \Delta h^{j}(t_0-T)}{T}
\end{align} 
where $\hat \alpha_2^j$ is referred to as the predicted rate in frequency, and $\Delta h^{j}(t)=h^{j}(t)-$UTC$(t)$, $t=kT$.
Thus, the time deviations during the operating interval $t\in(t_0,t_T)$ can be predicted by
\begin{align}\label{eq:JST1}
    \Delta \hat h^{j}(t_{k+1}) &= \Delta \hat h^{j}(t_{k}) + \hat \alpha_2^j \tau_k, \ \ j=1,\ldots,m. 
\end{align} 

\subsubsection{Weighting and updating procedure} \label{sec:Weighting}
In addition to the above prediction procedure, JST-algo includes  weighting and updating procedures associated with measurements to equalize the nonequal clock residuals $\epsilon_i(t):=\Delta h^i(t)-\Delta \hat h^i(t)$ to avoid discontinuities as much as possible.
This is because, without such a procedure, the discontinuity that appeared in the averaged atomic time when one of the clocks leaves the ensemble may give rise to instability \cite{hanado2006improvement}.  
In JST-algo,
the predicted values \eqref{eq:JST1} are to be modified as
\begin{align}\label{eq:JST2}
\Delta \hat h^{m}(t_{k+1}) 
&= \sum\nolimits_{i=1}^m \beta_i\left(\Delta \hat h^{i}(t_{k+1}) -y_{im}(t_{k+1})\right)
\\ \label{eq:JST3}
\Delta \hat h^{i}(t_{k+1}) 
&=  \Delta \hat h^{m}(t_{k+1})+y_{im}(t_{k+1}), \quad i\not=m
\end{align}
where  $\Delta\hat h^{i}(t_{k+1})$ in right-hand side of \eqref{eq:JST2} is understood the one in \eqref{eq:JST1}.

If there is no observation noise,  
the modified clock residuals after the 2-step procedure \eqref{eq:JST2} and \eqref{eq:JST3} 
are equalized as the averaged atomic time ${\rm TA}(t_{k+1})$ after the prediction procedure \eqref{eq:JST1} for all the clocks.
This can be verified as
\begin{align}\nonumber
   \epsilon_m(t_{k+1}) &= \Delta h^m(t_{k+1})-\!\sum\nolimits_{i=1}^m\! \beta_i\!\left(\Delta \hat h^{i}(t_{k+1})\!-\!y_{im}(t_{k+1})\right)
   \\ \nonumber
   &=\sum\nolimits_{i=1}^m \!\beta_i\left(\Delta h^{m}(t_{k+1})\!-\!\Delta \hat h^{i}(t_{k+1}) \!+\!y_{im}(t_{k+1})\right)\\  \nonumber
   &=\sum\nolimits_{i=1}^m \!\beta_i\left(\Delta h^{i}(t_{k+1})-\Delta \hat h^{i}(t_{k+1}) \right),
   \\ \nonumber
\epsilon_i(t_{k+1}) &= \Delta h^i(t_{k+1})-\Delta \hat h^{m}(t_{k+1})+y_{im}(t_{k+1})
  \\ \nonumber
   &=\Delta h^m(t_{k+1})-\Delta \hat h^{m}(t_{k+1})
     =\epsilon_m(t_{k+1}), \quad i\not=m.
\end{align}

\subsection{Algorithm of Averaged Atomic Time by Conventional Kalman Filter}\vspace{-1pt}
Besides JST-algo, there is another famous algorithm of averaged atomic time based on the Kalman filter. Specifically, 
the CKF-algo for the time scale generation in the operating interval $t\in(t_0,t_T)$ is given by
\begin{align} \label{eq:CKF1}
  & \bm K_k = \bm P_k {{\bm H} ^{\sf T}\bigr(\bm{H}\bm P_k{ \bm H}^{\sf T}+\hat R\bigr)^{-1}}  \\  \label{eq:Pk}
  &\bm P_{k+1}=    \bm F[k] (\bm P_k-\bm K_k {\bm H}\bm P_k) \bm F[k]^{\sf T} +\hat W  \\
  \label{eq:CKF2}
  &\hat{\bm x}[k+1] =\bm F[k] \hat{\bm x}[k] + \bm K_k(\bm{y}[k]-\bm{H}\hat{\bm x}[k]) 
\end{align}
with the initial $\bm P_0=pI_{nm}$ for some constant $p\in\mathbb{R}_+$ 
and the guess of the initial state 
$\hat{\bm{x}}[0]=(\hat{\bm{x}}_1^\tr[0],  \ldots, \hat{\bm{x}}_n^\tr[0])^\tr\approx{\bm{x}}[0]$, 
where $\bm K_k$ and $\bm P_k$ 
are the Kalman gain and error covariance, respectively, 
$\hat R$ and $\hat W$ are the guesses of the covariance in observation noise 
$ \bm{w}[k]$ and system noise $\bm{v}[k]$, respectively.
Then, the predicted time deviations  $\hat{\bm{x}}_1$ is hence expressed by
\begin{align}
\hat{\bm{x}}_1[k]= \bigl(C \otimes I_m\bigr)\hat{\bm{x}}[k].
\end{align}

\section{Main Results}\label{sec:moddes}

\subsection{Generalized JST-algo via State-Space Model}  

In this section, we present the generalized JST-algo for the atomic clock ensemble with higher-oder clocks, i.e., $n\geq2$, where the pseudocode of the generalized JST-algo is shown in Algorithm~\ref{Genearlized_JST} below.
The only difference between Algorithms~\ref{Hanado's loop} and~\ref{Genearlized_JST} is in that 
the prediction procedure \eqref{eq:JST1} for the time deviation ${\bm{x}}_1=(\Delta h^{1},\ldots,\Delta h^{m})^\tr$ is generalized as
\begin{equation}    
\simode{\hat{\bm{x}}[k+1] &= \bm F[k]\hat{\bm{x}}[k]\\
\label{eq:JST_new1}
\hat{\bm{x}}_1[k+1]&= \bigl(C \otimes I_m\bigr)\hat{\bm{x}}[k+1]
}
\end{equation}  
which is compatible with \eqref{eq:JST1} for the second-oder clocks since 
$\hat{\bm{x}}_2[k]=\hat{\bm{x}}_2[0]=(\hat{\alpha}_{2}^1,\ldots,\hat{\alpha}_{2}^m)^\tr$ stands for the set of predicted rate in frequency in \eqref{eq:motivation2} for any $k=0,1,\ldots,T$.

\begin{algorithm}[t]
\caption{Generalized JST-algo}
\begin{algorithmic}[1]
\State {\bf Initialization}: $\hat{\bm{x}}[0]\approx {\bm{x}}[0]$, and $k=1$
\While{$k\leq T$} 
\State $\hat{\bm{x}}[k] = \bm F[k-1]\hat{\bm{x}}[k-1]$  \Comment{Prediction}
\State $(\Delta \hat h^{1}[k],\ldots,\Delta \hat h^{m}[k])^\tr=\bigl(C \otimes I_m\bigr)\hat{\bm{x}}[k]$
\State $\Delta\hat{h}^m[k]=\sum\nolimits_{i=1}^m \beta_i\left(\Delta\hat{h}^i[k]-y_{im}[k]\right)$ \Comment{Weighting}
\For{$i\not=m$}
\State $\Delta\hat{h}^i[k]=\Delta\hat{h}^m[k]+y_{im}[k]$ \Comment{Update}
\EndFor
\State $k\leftarrow k+1$
\EndWhile
\end{algorithmic}\label{Genearlized_JST}
\end{algorithm}

In the state space form, the weighting procedure \eqref{eq:JST2} along with the prediction procedure \eqref{eq:JST_new1} can be expressed by
\begin{equation} 
\Delta \hat h^{m}[k+1] =\beta^{\tr}\Bigl\{
				   \bigl(C \otimes I_m\bigr)\bm F[k]\bm{x}[k] - e_{1:m-1} \bm{y}[k+1]
				   \Bigr\}
\end{equation}
and hence the individual predicted time deviation of the other clocks in \eqref{eq:JST2} is subsequently updated by
\begin{align}\nonumber
    \Delta \hat h^{i}[k+1] =& \Delta \hat h^{m}[k+1]+y_{im}[k+1]\\  \nonumber
  				   =& \beta^{\tr}\Bigl\{
				                       \bigl(C \otimes I_m\bigr)\bm F[k]\bm{x}[k] - e_{1:m-1} \bm{y}[k+1]
				                       \Bigr\}   \\
				       &+ y_{im}[k+1], \quad i\not=m.
\end{align}
As a result, the generalized JST-algo is expressed as
\begin{align} \nonumber
\hat{\bm{x}}_1[k+1] =
& \mathds{1}_m \beta^{\tr}
\Bigl\{
    \Bigl(C \otimes I_m\Bigr)\bm F[k]\hat{\bm{x}}[k] - e_{1:m-1} \bm{y}[k+1]
\Bigr\} \\ \label{eq:Hanado1_general}
&+ e_{1:m-1} \bm{y}[k+1],  \\ \label{eq:Hanado2_general}
\hat{\bm{x}}_{2:n}[k+1] =
& \bigl(  {A}_{2:n} [k] \otimes I_m \bigr) \hat{\bm{x}}_{2:n}[k]
 \end{align} 
where ${A}_{2:n}[k]$ is the dimension-reduced matrix of the matrix ${A} [k]$ by removing the first row and column.

\subsection{Equivalence of  The Generalized JST-algo and CKF-algo in Averaged Atomic Time}  
In this section, we reveal the equivalence between JST-algo and CKF-algo in averaged atomic time ${\rm TA}[k]$.
Specifically, we begin with a fundamental theoretical analysis of JST-algo.

\begin{theorem}   \label{prop:hanado}
Consider the system model \eqref{eq:Ndmodel} for an $m$-clock ensemble.
For a given initial guess $\hat{\bm x}[0]$, it follows that the averaged atomic time ${\rm TA}[k]$ of the generalized JST-algo
\begin{equation}  \label{eq:weq}
    {\rm TA}[k]= C \bm{\epsilon}_{\rm ens}[k]  , \quad k=0,1,\ldots, T,
\end{equation}
does not depend on the observation noise 
$\bm{w}[\cdot]$ 
for any weight $\beta$ 
satisfying $\beta_1+\ldots+\beta_m=1$.
In addition, the ensemble prediction error 
$\bm{\epsilon}_{\rm ens}$ 
is given by
\begin{equation} \label{eq:hatepsilon}
\bm{\epsilon}_{\rm ens}[k+1]   
= A[k] \bm{\epsilon}_{\rm ens}[k]
  +\bigl( I_n \otimes \beta^{\tr} \bigr) \bm{v}[k]
\end{equation}
 \end{theorem} 
 
\begin{proof} 
For the following analysis, we let 
\[
P:= \mathds{1}_m \beta^{\tr}
,\quad
\overline{P}:= I_m-\mathds{1}_m \beta^{\tr}
\]
which are projection matrices satisfying 
\[
\overline{P}=\overline{V}^{\dagger}\overline{V}-\mathds{1}_m(\beta-\tfrac{1}{m}\mathds{1}_m)^{\tr},\quad 
\overline{P}e_{1:m-1}\overline{V}=\overline{P}.
\]
Now letting $\overline{V}^{\ddag}:=\overline{V}^{\dagger}-\mathds{1}_m(\beta-\tfrac{1}{m}\mathds{1}_m)^{\tr}e_{1:m-1}$, 
the predicted time deviations of the JST-algo are written as
\begin{align} \nonumber
\!\!\!\!\hat{\bm{x}}_1[k+1]  
= & \bigl(C  {A}[k] \otimes P \bigr)\hat{\bm{x}}[k]+\overline{P} e_{1:m-1}\bm{y}[k+1]
\\ \nonumber
= & \bigl(C  {A}[k] \otimes P \bigr)\hat{\bm{x}}[k] 
    +\overline{P}e_{1:m-1} \bigl\{ \Bigl(C  A[k] \!\otimes\! \overline{V}\bigr)  \bm{x}[k] 
\\ \nonumber
   &\qquad+  \bigl(C   \otimes \overline{V}\bigr)\bm{v}[k]+\bm{w}[k+1]   \bigr\}
\\  \nonumber
= & \bigl(C  {A}[k] \otimes P \bigr)\hat{\bm{x}}[k] + \bigl(CA[k] \!\otimes\! \overline{P} \bigr)\bm{x}[k]
\\
   &+ \overline{P} \bm{v}_1[k] +\overline{V}^{\ddag}\bm{w}[k+1],
\end{align} 
Thus the prediction error $\bm{\epsilon}[\cdot]:=\bm{x}[\cdot] - \hat{\bm{x}}[\cdot]$ of JST-algo follows  
\begin{align}  \nonumber
\bm{\epsilon}_1 [k+1] 
&=\bigl(C {A}[k] \otimes P\bigr)\bm{\epsilon} [k]
    + P \bm{v}_1 [k] -\overline{V}^{\ddag}\bm{w}[k+1] 
\\
\bm{\epsilon}_{2:n} [k+1] 
&= \bigl( A_{2:n}[k] \otimes I_m \bigr) \bm{\epsilon}_{2:n}[k] 
    +\bm{v}_{2:n}[k].
\end{align} 
where $\bm{\epsilon}_1$ is understood as the clock residuals $\hat{\bm{x}}_1-\hat{\bm{x}}_1$ of time deviations.
Equivalently,  
we have
\begin{align}  \nonumber
\bm{\epsilon}[k+1] 
&=\underbrace{\mat{ &\!\!\!\!\!\!\!\!\!\!\!\!C {A}[k] \otimes P\\ 0 & A_{2:n}[k] \otimes I_m}}_{F^{\ddag}\bigl(A[k] \otimes I_m\bigr)}\bm{\epsilon} [k]
      -\underbrace{\mat{ \overline{V}^{\ddag} \\ 0}}_{\overline{F}}\bm{w}[k+1] \\  \label{eq:epsilon_JST}
  &  + \underbrace{\mat{&\!\!\!\!\!\!\!\!\!\!\!\!C  \otimes P\\ 0 &  I_{n-1}\otimes  I_m}}_{F^{\ddag}}  \bm{v}[k] .
\end{align} 
Here, note that $\overline{F}$, $F^{\ddag}$
satisfy
$\bigl( I_n \otimes \beta^{\tr} \bigr)\overline{F} =0$,
$\bigl( I_n \otimes \beta^{\tr} \bigr)F^{\ddag}= I_n \otimes \beta^{\tr} $ 
due to $\beta^{\tr}\overline{V}^{\ddag}=0$, $\beta^{\tr} P=\beta^{\tr}$  
and hence we obtain 
\begin{align}  \nonumber
\bm{\epsilon}_{\rm ens}[k+1] 
&=  
\bigl( I_n \otimes \beta^{\tr} \bigr)\bigl\{ 
                        F^{\ddag}\bigl(A[k] \otimes I_m\bigr)\bm{\epsilon} [k] \\ \nonumber
     & \qquad  +F^{\ddag}\bm{v}[k]
                      -\overline{F} \bm{w}[k+1]
 \bigr\}\\ \nonumber
 &=  \bigl( I_n \otimes \beta^{\tr} \bigr)\Bigl(A[k] \otimes I_m\Bigr)\bm{\epsilon} [k]
   +\bigl( I_n \otimes \beta^{\tr} \bigr) \bm{v}[k]\\  \label{eq:weigt1}
&= A[k] \bigl( I_n \otimes \beta^{\tr} \bigr)\bm{\epsilon} [k]+\bigl( I_n \otimes \beta^{\tr} \bigr) \bm{v}[k] ,
\end{align}
which completes proof.
\end{proof} 

\begin{remark}   \label{rem:1} 
The result of Theorem~\ref{prop:hanado}  can contribute to theoretically explaining why NMIs all over the world often require more atomic clocks to generate more accurate time scales. 
Specifically, consider a homogeneous ensemble with a constant sampling interval $\tau_k=\tau$ for $k=0,1,\ldots,T$, 
where the covariance of state noise $\bm{v}[k]$ is written as $Q\otimes I_m$ for some $Q$.
In this case, supposing that the initial error $\bm{\epsilon}[0]= \hat\mu_0 \otimes \mathds{1}_m$ for some $\hat\mu_0\in\mathbb R^n$ and the initial covariance of prediction error is $\bm P_0$, 
it can be shown from \eqref{eq:weq} that the expected value of the averaged atomic time
$
 \mathbb E [{\rm TA} [k ]] =CA^k(\tau)\hat\mu_0
$
does not depend on the number $m$ of the atomic clocks,
but the covariance
 \begin{align}    \nonumber
 \mathbb C({\rm TA} [k ])= &C\Big\{A^k(\tau) \bigl( I_n \otimes \beta^{\tr} \bigr) \bm P_0\bigl( I_n \otimes \beta\bigr)  A^k(\tau)^{\tr}\\  
 &\qquad+\sum\nolimits_{i=0}^{k-1}A^i(\tau)\beta^{\tr}\beta QA^i(\tau)^{\tr}\Big\}C^{\tr}
\end{align}
is diminished by increasing the number $m$ of the atomic clocks when the weights of the clocks are set to all the same (since $\beta^{\tr}\beta=\tfrac{1}{m}$ under $\beta=\tfrac{1}{m}\mathds{1}_m$).
In other words, better prediction performance of the averaged atomic time can be achieved for the atomic clock ensemble with larger number $m$ of the clocks.
\end{remark} 

Now, we begin to make a theoretical analysis for CKF-algo.
It is well known that the CKF-algo may result in numerical instability in the real implementation of time generations.
This is because the computation error accumulates in the Kalman gains due to the divergence of error covariance under undetectability of the state space model. 
Two covariance reduction methods are developed in \cite{brown1991theory} and \cite{greenhall2006kalman}
to suppress the numerical instability in the real implementation.
However, neither of them can absolutely avoid the appearance of numerical instability especially for the case when the initial covariance $\bm P_0$ is large, and the theoretical analysis of CKF-algo for the ideal case without computation errors is still unclear.

%\begin{figure}
%\centering
%\includegraphics[width = .90\linewidth]{metri_instability-eps-converted-to.pdf}
%\caption{Example of numerical instability of a 4-clock ensemble under CKF-algo in real implemtation.}
%\label{metri_instability}
%\end{figure}

To reveal the theoretical expression of averaged atomic time ${\rm TA}[k]$ of the CKF-algo, 
we note that the state profile $\bm x$ of the ensemble can be decomposed by observable Kalman canonical decomposition as
\begin{equation} \label{eq:tranformation}
\bm x=\mat{I_n\otimes \overline{V}^{\dagger}&I_n\otimes\mathds{1}_m}\mat{ \bm \xi_{\rm o}   \\  \bm{\xi}_{\rm \bar o}  } 
\end{equation}
where $\bm \xi_{\rm o}:=\bigl(I_n \otimes\overline{V})\bm x\in\mathbb R^{nm-n}$, 
$\bm{\xi}_{\rm \bar o}:=\tfrac{1}{m}\bigl(I_n \otimes\mathds{1}_m^{\tr})\bm x\in\mathbb R^{n}$ 
denote the observable and unobservable state, respectively. 
Using this fact and letting \begin{align}
&\bm{\epsilon}_{\rm o}[k]
:={\bm{\xi}}_{\rm o}[k]-\hat{\bm{\xi}}_{\rm o}[k]
=\bigl(I_n \otimes\overline{V})\bm\epsilon[k], \\
&\bm{\epsilon}_{\rm \bar o}[k]
:={\bm{\xi}}_{\rm \bar o}[k]-\hat{\bm{\xi}}_{\rm \bar o}[k]
=\tfrac{1}{m}\bigl(I_n \otimes\mathds{1}_m^{\tr})\bm\epsilon[k],
\end{align}
the ensemble prediction error
$\bm{\epsilon}_{\rm ens}$  is transformed as
\begin{align} \nonumber
\bm{\epsilon}_{\rm ens}[k]
&=(I_n \otimes \beta^{\tr})\mat{I_n\otimes \overline{V}^{\dagger}&I_n\otimes\mathds{1}_m}\mat{ \bm \epsilon_{\rm o} [k]  \\  \bm{\epsilon}_{\rm \bar o}[k]  } \\ \label{eq:decom_ens}
&=\bigl(I_n\otimes \beta^{\tr}\overline{V}^{\dagger}\bigr)\bm{\epsilon}_{\rm o}[k]
+\bm{\epsilon}_{\rm \bar o} [k].
\end{align}
That is to say, once we derive the dynamics of $\bm{\epsilon}_{\rm o}[k]$ and $\bm{\epsilon}_{\rm \bar o}[k]$, 
the dynamics of $\bm{\epsilon}_{\rm ens}[k]$ can be accordingly derived. 

In terms of CKF-algo,
it can be theoretically shown that the predicted observable state $ \hat{\bm \xi}_{\rm o}$ under CKF-algo follows
\begin{align} \label{eq:Kalmangain}
&\hat{\bm K}_k = \hat{\bm P}_k {\bm{H}_{\rm o} ^{\sf T}\bigr(\bm{H}_{\rm o}\hat{\bm P}_k\bm{ H}_{\rm o}^{\sf T}+\hat R\bigr)^{-1}}  \\   \label{eq:Lk_PKF}
&\hat{\bm P}_{k+1}=    \bm F_{\rm o}[k] (\hat{\bm P}_k -\hat{\bm K}_k \bm{H}_{\rm o}\hat{\bm P}_k) \bm F_{\rm o}[k] ^{\sf T} +W_{\rm o}  
\\ \label{eq:PKA2_o}
&\hat{\bm{\xi}}_{\rm o}[k+1] =\bm F_{\rm o}[k]  \hat{\bm{\xi}}_{\rm o}[k] +\hat{\bm K}_k(\bm{y}[k]-\bm{H}_{\rm o} \hat{\bm{\xi}}_{\rm o}[k]) 
\end{align}
where $ \bm F_{\rm o}[k] := A[k]\otimes I_{m-1}$, $\bm H_{\rm o}:=C \otimes I_{m-1}$, $W_{\rm o}:=  ( I_n \otimes \overline{V})\hat W( I_n \otimes \overline{V})^{\tr}$, are the system matrix, measurement matrix, and system noise covariance for the observable subspace.
Note that $ \hat{\bm K}_k:=( I_n \otimes \overline{V})\bm K_k$ and $ \hat{\bm P}_k:=  ( I_n \otimes \overline{V}) \bm P_k( I_n \otimes \overline{V})^{\tr}$
are understood as the Kalman gain and the error covariance of CKF-algo in observable state space, respectively. 
The detailed derivation of \eqref{eq:PKA2_o} is attached in Appendix below for reference. 

The next result shows that the averaged atomic time ${\rm TA}[k]$  of JST-algo and CKF-algo are equivalent to each other if we adopt some specific guesses of the noise covariance and choose equal weights for the clocks.

\begin{theorem} \label{prop:partial_kalman_obse}
Consider the system model \eqref{eq:Ndmodel}  for an $m$-clock ensemble.
For a given initial guess $\hat{\bm x}[0]$,
if the guess $\hat W$ of the system noise covariance is given by $\hat W=Q\otimes I_m$
for some $Q\geq0$, 
then the averaged atomic time ${\rm TA}[k]$ of the CKF-algo is given by \eqref{eq:weq}
with the ensemble prediction error
\begin{align}\nonumber
\bm{\epsilon}_{\rm ens}[k+1]   
=& A[k] \bm{\epsilon}_{\rm ens}[k]
  +( I_n \otimes \beta^{\tr} ) \bm{v}[k]\\ \label{eq:CKF_ens}
  &-( I_n \otimes \beta^{\tr} \overline{V}^{\dagger} )\hat{\bm K}_k (\bm{H}_{\rm o}\bm{\epsilon}^{\rm CKF}_{\rm o} [k]+ \bm w[k])
\end{align}
where the prediction error of observable state $\bm \xi_{\rm o}$ follows
\begin{align} \nonumber
\bm{\epsilon}^{\rm CKF}_{\rm o}[k+1]=&(\bm F_{\rm o}[k] -\hat{\bm K}_k\bm{H}_{\rm o})\bm{\epsilon}^{\rm CKF}_{\rm o}[k]-\hat{\bm K}_k\bm w[k]
\\  \label{eq:CKF_ob}
&+( I_n \otimes \overline{V})\bm v[k] .
\end{align}
Furthermore, with such a $\hat W$,
the generalized JST-algo and CKF-algo generate 
the same averaged atomic time ${\rm TA}[k]$ for $k=0,1,\ldots, T$
if and only if the weights of the atomic clocks are all equal, i.e.,
$\beta= \tfrac{1}{m} \mathds{1}_m$.
\end{theorem} 

\begin{proof} 
First, using  
$
\bm{y}[k] = \bm H \bm{x}[k] + \bm{w}[k]=\bm H_{\rm o} \bm{\xi}_{\rm o}[k] + \bm{w}[k]
$,
it follows that the derivation \eqref{eq:CKF_ob} of the prediction error in the observable state is immediate from \eqref{eq:PKA2_o} since $\bm \xi_{\rm o} $ follows 
\[
\bm \xi_{\rm o}[k+1]=\bm F_{\rm o}[k] \hat{\bm \xi}_{\rm o}[k] + ( I_n \otimes \overline{V})\bm v[k] .
\]
In terms of the prediction error in the unobservable state, 
since the prediction error of Kalman filter is   
$
\bm{\epsilon} [k+1]=(\bm F[k]-\bm K_k\bm{H})\bm{\epsilon}[k]-\bm K_k\bm w[k]+\bm v[k] ,
$
we have  
\begin{align}\nonumber
\bm{\epsilon}^{\rm CKF}_{\rm \bar o}[k+1] 
=&A[k]\bm{\epsilon}^{\rm CKF}_{\rm \bar o}[k]+\tfrac{1}{m}\bigl(I_n \otimes\mathds{1}_m^{\tr})\bm{v}[k]   
\\  \label{eq:weighted_error}
&-\tfrac{1}{m}\bigl(I_n \otimes\mathds{1}_m^{\tr})\bm K_k(\bm{H} \bm{\epsilon}[k]+ \bm w[k])
\end{align} 
Note that the condition $\hat W=Q\otimes I_m$ indicates 
$(I_n \otimes\mathds{1}_m^{\tr})\bm K_k=(I_n \otimes\mathds{1}_m^{\tr})\bm P_k {{\bm H}^{\sf T}\bigr(\bm{H}\bm P_k{ \bm H}^{\sf T}+\hat R\bigr)^{-1}}=0$ for $k=0,1,\ldots, T$, because $(I_n \otimes\mathds{1}_m^{\tr})\bm P_k {\bm H} ^{\sf T}=0$, $k=0,1,\ldots, T$.
In particular, since $\bm P_0=pI_{nm}$, it follows that
$(I_n \otimes\mathds{1}_m^{\tr})\bm P_0 {\bm H} ^{\sf T}=0$ holds.
Furthermore, let $(I_n \otimes\mathds{1}_m^{\tr})\bm P_k {\bm H} ^{\sf T}:=\bm U_k$, we have
\begin{align}\nonumber
\bm U_1=&(I_n\otimes \mathds{1}_m^{\tr})
\left(\bm F[0] (\bm P_0-\bm K_0 {\bm H}\bm P_0) \bm F[0]^{\sf T} +\hat W\right) {\bm H}^{\sf T}
\\ \nonumber
=&pA[0]A[0]^{\sf T}(I_n\otimes \mathds{1}_m^{\tr}){\bm H}^{\sf T}
+(I_n\otimes \mathds{1}_m^{\tr})\hat W{\bm H}^{\sf T}=0,
\\ \nonumber
\bm U_2=&(I_n\otimes \mathds{1}_m^{\tr})
\left(\bm F[1] (\bm P_1-\bm K_1 {\bm H}\bm P_1) \bm F[1]^{\sf T} +\hat W\right) 
{\bm H}^{\sf T}\\ \nonumber
=&A[1](I_n\otimes \mathds{1}_m^{\tr})\left(\bm F[0] (\bm P_0-\bm K_0 {\bm H}\bm P_0) \bm F[0]^{\sf T} +\hat W\right) 
\\  \nonumber
&\cdot\bm F[1]^{\sf T} {\bm H}^{\sf T}+(I_n\otimes \mathds{1}_m^{\tr})\hat W{\bm H}^{\sf T}
\\ \nonumber
=&pA[1]A[0]A[0]^{\sf T}A[1]^{\sf T}(I_n\otimes \mathds{1}_m^{\tr}){\bm H}^{\sf T}
=0,
\end{align}
whereas the proof for $\bm{U}_k=0$, $k>2$ can be similarly handled.

Thus, the prediction error of unobservable state is given by
\begin{align} \label{eq:CKF_obar}
\bm{\epsilon}^{\rm CKF}_{\rm \bar o}[k+1] 
=&A[k]\bm{\epsilon}^{\rm CKF}_{\rm \bar o}[k]+\tfrac{1}{m}\bigl(I_n \otimes\mathds{1}_m^{\tr})\bm{v}[k]   
\end{align} 
Now, noting that 
$\beta^{\tr}\overline{V}^{\dagger}\overline{V}=\beta^{\tr}-\tfrac{1}{m}\mathds{1}_m^{\tr}=\beta^{\tr}_{\rm df}$,
it follows from
\begin{align} \nonumber
\bm{\epsilon}_{\rm ens}[k\!+\!1]
=&\bigl(I_n\otimes \beta^{\tr}\overline{V}^{\dagger}\bigr)\bm{\epsilon}^{\rm CKF}_{\rm o}[k+1]
+\bm{\epsilon}^{\rm CKF}_{\rm \bar o} [k+1]
\\ \nonumber
=
&\bigl(I_n\otimes \beta^{\tr}\overline{V}^{\dagger}\bigr)
\Big(\bm F_{\rm o}[k]\bm{\epsilon}^{\rm CKF}_{\rm o}[k] +( I_n \otimes \overline{V})\bm v[k] \Big)
\\ \nonumber
&+A[k]\bm{\epsilon}^{\rm CKF}_{\rm \bar o}[k]+\tfrac{1}{m}\bigl(I_n \otimes\mathds{1}_m^{\tr})\bm{v}[k]  
\\ \nonumber
&-\bigl(I_n\otimes \beta^{\tr}\overline{V}^{\dagger}\bigr)\hat{\bm K}_k
\Big(\bm{H}_{\rm o}\bm{\epsilon}^{\rm CKF}_{\rm o}[k]+\bm w[k])\Big)
\\ \nonumber
=
&A[k]\bigl(I_n\otimes \beta^{\tr}\bigr)\bm{\epsilon}[k]+\bigl(I_n\otimes \beta^{\tr}\bigr)\bm{v}[k]
\\ 
&-\bigl(I_n\!\otimes\! \beta^{\tr}\overline{V}^{\dagger}\bigr)\hat{\bm K}_k
\Big(\bm{H}_{\rm o}\bm{\epsilon}^{\rm CKF}_{\rm o}[k]\!+\!\bm w[k])\Big)
\end{align}
that \eqref{eq:CKF_ens} holds. 
Then it can be seen
that CKF-algo and JST-algo generate the same averaged atomic time ${\rm TA}[k]$ if and only if 
$
\bigl(I_n\otimes \beta^{\tr}\overline{V}^{\dagger}\bigr)\hat{\bm K}_k=\bigl(I_n\otimes \beta^{\tr}\overline{V}^{\dagger}\overline{V}\bigr){\bm K}_k=(I_n\otimes \beta_{\rm df}^{\tr})\bm K_k=0
$, i.e., 
$(I_n \otimes\beta_{\rm df}^{\tr})\bm P_k {\bm H} ^{\sf T}=0$, $k=0,1,\ldots, T$.

Recalling $\bm P_0=pI_{nm}$ and $\beta_{\rm df}$ satisfies $\beta_{\rm df}^{\tr}\mathds{1}_m=0$, it follows that
$(I_n \otimes\beta_{\rm df}^{\tr})\bm P_0 {\bm H} ^{\sf T}=0$ and 
$(I_n \otimes\beta_{\rm df}^{\tr})\bm P_1 {\bm H} ^{\sf T}=0$ holds if and only if $\beta_{\rm df}=0$,
i.e., $\beta=\tfrac{1}{m} \mathds{1}_m$, 
which completes the proof. \end{proof}

\begin{remark}
The equivalence result in Theorem~\ref{prop:partial_kalman_obse} indicates that more precise averaged atomic time can be achieved by the clock ensemble with larger number $m$ of the clocks for both JST-algo and CKF-algo in a homogeneous ensemble (see Remark~\ref{rem:1}).  
However, it is worth noting that the computation cost of the CKF-algo may be larger than the generalized JST-algo when the number $m$ of the clocks is too big because CKF-algo requires more matrix computation than the generalized JST-algo. 
A discussion in terms of the runtime of CKF-algo and the generalize JST-algo will be given in Section~\ref{sec:example} later. 
\end{remark}

\subsection{Relation Between The Generalized JST-algo and CKF-algo in Clock Residual}  
%Even though the prediction error of JST-algo in observable state space is worse than CKF-algo in the sense that $\bm{\epsilon}^{\rm JST}_{\rm o}$ is diverging but $\bm{\epsilon}^{\rm CKF}_{\rm o}$ is converging, t

Except the comparison of accuracy with the averaged atomic time ${\rm TA}[k]$, it is important to compare  
clock residuals $\epsilon_i[k]:=\Delta h^i[k]-\Delta \hat h^i[k]$ of the CKF-algo and the generalized JST-algo because   
the clock residuals are related to stability of the generated time.
The next result reveals that the mean of the clock residual $\bm{\epsilon}_1=(\epsilon_1,\ldots,\epsilon_m)^\tr$ of  the generalized JST-algo and CKF-algo may be eventually equivalent in the case if  equal weights are considered for the clocks.

\begin{theorem}  \label{prop:thm3}
Consider the system model \eqref{eq:Ndmodel}  for an $m$-clock ensemble.
For a given initial guess $\hat{\bm x}[0]$,
if the weights of the atomic clocks are all equal, i.e.,
$
\beta= \tfrac{1}{m} \mathds{1}_m
$,
and if the guess $\hat W$ of the state noise covariance is given by $\hat W=Q\otimes I_m$,
for some $Q\geq0$, 
then the clock residuals $\bm{\epsilon}_1$ of the generalized JST-algo and CKF-algo respectively follow
\begin{align}  \label{eq:weqs}
    \bm{\epsilon}^{\rm JST}_1[k]&= 
    -\overline{V}^{\dagger}\bm{w}[k]
    +C \bm{\epsilon}^{\rm CKF}_{\rm \bar o} [k]\mathds{1}_m \\   \label{eq:erro_ind_CKF}
        \bm{\epsilon}^{\rm CKF}_1[k]&=  \overline{V}^{\dagger} \bm H_{\rm o}\bm{\epsilon}^{\rm CKF}_{\rm o}[k]
    +C \bm{\epsilon}^{\rm CKF}_{\rm \bar o} [k]\mathds{1}_m
\end{align}
for $k=1,2,\ldots, T$, where $\bm{\epsilon}^{\rm CKF}_{\rm o} [k]$ and $\bm{\epsilon}^{\rm CKF}_{\rm \bar o} [k]$ are  given by \eqref{eq:CKF_ob} and \eqref{eq:CKF_obar}, respectively.

Furthermore, 
if the sampling interval is constant, i.e., $\tau_k=\tau$ for $k=0,\ldots,T$,  
then the mean of  clock residual $\bm{\epsilon}_1$ satisfy
\begin{equation}  \label{eq:weqsas}
\lim_{k\to \infty}
\left\{
\mathbb E\Big[\bm{\epsilon}_1^{\rm JST}[k]
-\bm{\epsilon}_1^{\rm CKF}[k]\Big]
\right\}=0.
\end{equation}
If, in addition, there is no observation noise, i.e., $\bm{w}[k]=0$, $k=0,\ldots, T$, 
then the covariance of  clock residual $\bm{\epsilon}_1$ satisfy 
\begin{equation}  
\mathbb C\Big[\bm{\epsilon}_1^{\rm JST}[k]\Big]-\mathbb C\Big[\bm{\epsilon}_1^{\rm CKF}[k]\Big] \leq 0,\quad k=0,1,\ldots,T.
\end{equation} 

\begin{proof} 
First, similar to the decomposition of $\bm{\epsilon}_{\rm ens}$ in \eqref{eq:decom_ens}, 
the clock residual $\bm{\epsilon}_1$ can be decomposed into
\begin{align} \nonumber
\bm{\epsilon}_{1}[k]
&=(C \otimes I_m)\mat{I_n\otimes \overline{V}^{\dagger}&I_n\otimes\mathds{1}_m}\mat{ \bm \epsilon_{\rm o} [k]  \\  \bm{\epsilon}_{\rm \bar o}[k]  } \\ \nonumber
&=  \bigl(C\otimes \overline{V}^{\dagger}\bigr)\bm{\epsilon}_{\rm o}[k]
    +\bigl(C\otimes\mathds{1}_m\bigr)\bm{\epsilon}_{\rm \bar o} [k]\\ 
&=\overline{V}^{\dagger}\bigl(\underbrace{C\otimes  I_{m-1}}_{\bm H_{\rm o}}\bigr)\bm{\epsilon}_{\rm o}[k]+C\bm{\epsilon}_{\rm \bar o} [k]\mathds{1}_m.
\end{align}
which directly proves \eqref{eq:erro_ind_CKF} for CKF-algo.
Now it follows from the error dynamics \eqref{eq:epsilon_JST} of  the generalized JST-algo that
the prediction error $\bm{\epsilon}^{\rm JST}_{\rm o}=\bigl(I_n \otimes\overline{V})\bm\epsilon[k]$ of the generalized JST-algo follows 
\begin{align}   \nonumber
&\bm{\epsilon}^{\rm JST}_{\rm o}[k+1]=
\bigl(I_n \otimes\overline{V})\bigl\{ F^{\ddag}\bigl(A[k] \otimes I_m\bigr)\bm{\epsilon} [k]\\ \nonumber
&
\qquad\qquad\qquad+F^{\ddag}\bm{v}[k]-\overline{F} \bm{w}[k+1]\bigr\}\\ \nonumber
&=\bigl(A_{2:n}^0[k]\otimes \overline{V}\bigr)\bm{\epsilon}[k]
+\bigl(I^0_{n-1}\otimes \overline{V}\bigr)\bm{v}[k]\!-\!\mat{\!I_{m-1}\! \\ 0}\bm{w}[k+1]\\ \nonumber
&=\bigl(A_{2:n}^0[k]\otimes I_{m-1}\bigr)\bm{\epsilon}^{\rm JST}_{\rm o}[k]+\bigl(I^0_{n-1}\otimes \overline{V}\bigr)\bm{v}[k]\\    \label{eq:observable_JST}
&\quad-\mat{I_{m-1} \\ 0}\bm{w}[k+1]
\end{align}
which is diverging since the eigenvalue $\lambda$ of $A_{2:n}^0$ is given by $\lambda\in\{0,1,\ldots,1\}$.
However, note that  \eqref{eq:observable_JST} indicates 
\begin{equation}
%\bm H\bm{\epsilon}[k]=&(C\otimes I_{m-1})\bm{\epsilon}^{\rm JST}_{\rm o}[k]=-\bm{w}[k] 
%\\
\bm{H}_{\rm o}\bm{\epsilon}^{\rm JST}_{\rm o}[k+1]=-\bm{w}[k+1].
\end{equation}
Then, it follows from Theorem~\ref{prop:partial_kalman_obse} that 
\eqref{eq:weqs} is immediate since $\bm{\epsilon}^{\rm JST}_{\rm \bar o} [k]=\bm{\epsilon}^{\rm CKF}_{\rm \bar o} [k]$ under $\beta= \tfrac{1}{m} \mathds{1}_m$.
Next, under the condition $\tau_k=\tau$, note that the mean of $\bm{\epsilon}^{\rm CKF}_{\rm o}$ in \eqref{eq:CKF_ob} is converging to zero
because of the observable pair $(\bm F_{\rm o}[k],\bm{H}_{\rm o})$. 
Thus, \eqref{eq:weqsas} is immediate since
\begin{align}  \nonumber
\mathbb E\Big[\bm{\epsilon}_1^{\rm JST}[k]-\bm{\epsilon}_1^{\rm CKF}[k]\Big]
&=-\mathbb E\Big[\overline{V}^{\dagger}\bm{w}[k]
+\overline{V}^{\dagger} \bm H_{\rm o}\bm{\epsilon}^{\rm CKF}_{\rm o}[k]\Big]\\ 
&=-\overline{V}^{\dagger} \bm H_{\rm o}\mathbb E\Big[\bm{\epsilon}^{\rm CKF}_{\rm o}[k]\Big]
\end{align}
is converging to 0.  
Now, since
\begin{align}  \nonumber
\mathbb C\Big[\bm{\epsilon}_1^{\rm JST}[k]\Big]-\mathbb C\Big[\bm{\epsilon}_1^{\rm CKF}[k]\Big] 
&=-\mathbb C\Big[\overline{V}^{\dagger} \bm H_{\rm o}\bm{\epsilon}^{\rm CKF}_{\rm o}[k]\Big]\\ \nonumber
&=-\overline{V}^{\dagger} \bm H_{\rm o}\hat{\bm P}[k]\bm H_{\rm o}^{\tr}\bigl(\overline{V}^{\dagger} \bigr)^{\tr}\!\leq0
\end{align}
holds for $\bm{w}[k]=0$, $k=0,1,\ldots,T$,  the proof is complete.
\end{proof}
\end{theorem} 

\begin{remark}   
Note that the term $C \bm{\epsilon}^{\rm CKF}_{\rm \bar o} [k]$ in both \eqref{eq:weqs} and  \eqref{eq:erro_ind_CKF} under the generalized JST-algo and CKF-algo is nothing but the averaged atomic time ${\rm TA}[k]$ since the prediction error $\bm{\epsilon}^{\rm CKF}_{\rm \bar o}$ of observable state in \eqref{eq:CKF_obar} reduces to the ensemble prediction error
$\bm{\epsilon}_{\rm ens}$ in \eqref{eq:hatepsilon} when the conditions of Theorem~\ref{prop:thm3} are satisfied. 
Therefore, the result \eqref{eq:weqs} in Theorem~\ref{prop:thm3} indicates that if there is no observation noise, i.e., $\bm{w}[k]=0$, $k=0,\ldots, T$,  then the clock residuals under the generalized JST-algo are equalized to the averaged atomic time ${\rm TA}[k]$ for all the clocks, which is consistent with the analysis in Section~\ref{sec:Weighting}. 
Meanwhile,  the result also indicates that sophisticated measuring equipment (hardware) is required in the implementation of the generalized JST-algo to guarantee stability of the generated local time (otherwise the clock residuals may be significantly different to each other and hence instability may be risen when one of the clocks leaves the ensemble). 
\end{remark}   

\begin{remark}   
Theorem~\ref{prop:thm3} indicates that JST-algo can guarantee lower covariance of the clock residual $\bm{\epsilon}_1$ than CKF-algo when the observation noise is small enough. 
More precisely, if the actual covariance $R$ of the observation noise $\bm{w}[k]$ satisfies 
\begin{align}   \label{eq:condition_R}
R-\bm H_{\rm o}\hat{\bm P}_{\rm ss}\bm H_{\rm o}^{\tr}<0
\end{align}
then the covariance of the clock residuals $\bm{\epsilon}_1[k]$ of the generalized JST-algo is never larger than CKF-algo as $k\to \infty$ due to
\begin{align}  \nonumber
&\lim_{k\to \infty}
\left\{\mathbb C\Big[\bm{\epsilon}_1^{\rm JST}[k]\Big]-\mathbb C\Big[\bm{\epsilon}_1^{\rm CKF}[k]\Big] 
\right\}
\\ \nonumber
&=\lim_{k\to \infty}\left\{\mathbb C\Big[-\overline{V}^{\dagger}\bm{w}[k]\Big]
-\mathbb C\Big[\overline{V}^{\dagger} \bm H_{\rm o}\bm{\epsilon}^{\rm CKF}_{\rm o}[k]\Big]\right\}
\\   \label{eq:condition_R_skech}
&=\overline{V}^{\dagger}\Big(R- \bm H_{\rm o}\hat{\bm P}_{\rm ss}\bm H_{\rm o}^{\tr}\Big)\bigl(\overline{V}^{\dagger} \bigr)^{\tr}
\leq 0.
\end{align}
where $\hat{\bm P}_{\rm ss}$ is the steady-state covariance of \eqref{eq:Lk_PKF}
satisfying 
the algebraic Riccati equation given by 
\begin{align}  \nonumber
0=
&-\bm F_{\rm o}[k] \hat{\bm P}_{\rm ss}\bm{H}_{\rm o}^{\sf T}
\bigr(\bm{H}_{\rm o}\hat{\bm P}_{\rm ss}{\bm H}_{\rm o}^{\sf T}+\hat R\bigr)^{-1}
{\bm H}_{\rm o}\hat{\bm P}_{\rm ss}\bm F_{\rm o} [k] ^{\sf T}\\ \label{eq:Riccati}
&+\bm F_{\rm o}[k] \hat{\bm P}_{\rm ss}\bm F_{\rm o}[k] ^{\sf T}-\hat{\bm P}_{\rm ss}+W_{\rm o}.
 \end{align}
 In such a case, combining the result of Theorems~\ref{prop:partial_kalman_obse} and~\ref{prop:thm3}, the generalized JST-algo is hence considered as a better algorithm than CKF-algo for homogeneous ensembles.
 This is because in such a case, the clock residual $\epsilon_1$ of clock 1 satisfies  
\begin{align}  \nonumber
&\lim_{k\to \infty}
\left\{\mathbb C\Big[\epsilon_1^{\rm JST}[k]\Big]-\mathbb C\Big[\epsilon_1^{\rm CKF}[k]\Big] 
\right\}\\ \nonumber
&=\lim_{k\to \infty}e_i
\left(\mathbb C\Big[\bm{\epsilon}_1^{\rm JST}[k]\Big]\!-\!\mathbb C\Big[\bm{\epsilon}_1^{\rm CKF}[k]\Big] 
\right) 
e_i^{\tr}
\\  \label{eq:indival_better}
&=e_i\overline{V}^{\dagger}\Big(R- \bm H_{\rm o}\hat{\bm P}_{\rm ss}\bm H_{\rm o}^{\tr}\Big)\bigl(\overline{V}^{\dagger} \bigr)^{\tr}e_i^{\tr}\leq0,
\end{align}
where $e_i=[e_i^j]_{j=1,\ldots,m}$ denotes standard basis given by $e_i^i=1$,  $e_i^j=0$, $j\not=i$, e.g., $e_1=[1\ 0\ \ldots \ 0]$.
\end{remark}

In the case if the covariance $R$ of the observation noise $\bm{w}[k]$ is too large to satisfy the condition
\eqref{eq:condition_R}, the next result can be used to further compare the variance of clock residual of a specific clock 
using the steady-state covariance.

\begin{theorem}  \label{prop:thm4}
Consider the system model \eqref{eq:Ndmodel}  for an $m$-clock ensemble.
For a given initial guess $\hat{\bm x}[0]$,
if the conditions of Theorem~\ref{prop:thm3} are all satisfied but with possible non-zero observation noises, i.e., $R\geq0$, 
then the variance of  residual $\epsilon_i$ of clock $i$ of the generalized JST-algo and CKF-algo satisfy
\begin{equation}  
\lim_{k\to \infty}
\left\{\mathbb C\Big[\epsilon_i^{\rm JST}[k]\Big]-\mathbb C\Big[\epsilon_i^{\rm CKF}[k]\Big]
\right\} <0
\end{equation} 
if and only if the covariance $R$ of the observation noise $\bm{w}[k]$ satisfies 
\begin{align}   \label{eq:indival_better_cond}
\mathcal L_i:=e_i\overline{V}^{\dagger}\Big(R- \bm H_{\rm o}\hat{\bm P}_{\rm ss}\bm H_{\rm o}^{\tr}\Big)\bigl(\overline{V}^{\dagger} \bigr)^{\tr}e_i^{\tr}<0.
\end{align}
\begin{proof} 
The result is a direct consequence of Theorem~\ref{prop:thm3} with \eqref{eq:condition_R_skech} and  \eqref{eq:indival_better}.
\end{proof} 
\end{theorem}

\section{Numerical Simulations} \label{sec:example}
This section provides a couple of examples to demonstrate our results.
In particular, we  use a 5-clock ensemble (Example 1) with homogeneous second-order clocks to verify equivalence result in Theorem~\ref{prop:partial_kalman_obse}, and  use a 3-clock ensemble (Example 2) with third-order clocks to verify the result of Theorems~\ref{prop:thm3}  and  \ref{prop:thm4}  in terms of clock residuals.
\vspace{-5pt}
\subsection{Example 1: Second-order Clocks} 
Consider a second-order homogeneous atomic clock ensemble with $m=5$ clocks where variances of the system noises are set to  
 $\sigma_1^j=2.0587e-20$,  $\sigma_2^j=4.0760e-28$ for all the clocks.  
The sampling period is set to $\tau=0.1$s.
The variances of observation noises are set to $1e-12$ for all the clocks, i.e., $ R=1e-12I_{4}$.
In the simulation, the initial state $\bm{x}[0]$ of this 5-clock ensemble is set to a deterministic value around ${x}_{i}^j[0]\in (1e-15,2e-15)$.
The initial predicted value is set to $\hat{\bm{x}}[0]=1e-15\mathds{1}_5$. 
The guess of the  state noise (resp., measurement) covariance is set to the same as the actual one satisfying $\hat W=Q\otimes I_m$
for some $Q\geq0$ (resp., $\hat R=R$).
We let $T=36000$ so that we can discuss the performance of the algorithms in one hour.

\begin{figure}
 \centering
\includegraphics[width=8.4cm]{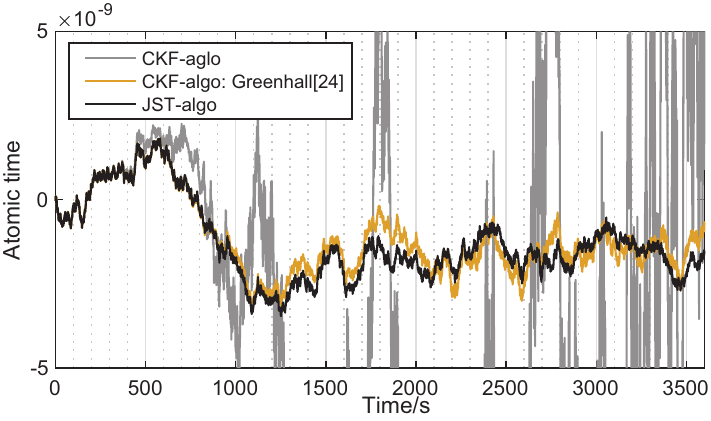}
\caption{The averaged atomic time ${\rm TA}[k]$ under JST-algo, CKF-algo, and CKF-algo \cite{greenhall2006kalman} with equal weights.}\label{TII_example1_different_initial} \label{TA_1}\vspace{-6pt}
\end{figure}

\subsubsection{Equal Wights}
In the case of equal weights for the clocks, i.e., $\beta=\tfrac{1}{5} \mathds{1}_5$,
the averaged atomic time  ${\rm TA}[k]$ of JST-algo is illustrated as the black line in Fig.~\ref{TA_1}
where the one simulated by CKF-algo with $\bm P_0=1e-8I$ is shown as the grey line. 
In this case, it follows from Theorem~\ref{prop:partial_kalman_obse} that CKF-algo and JST-algo ideally generate the same averaged atomic time  ${\rm TA}[k]$ at least if there are no calculation errors. 
The averaged atomic time  ${\rm TA}[k]$  of CKF-algo with the covariance reduction method \cite{greenhall2012review} is illustrated as the yellow line in Fig.~\ref{TA_1}.
It can be seen from this figure that the covariance reduction method is able to suppress the behavior of numerical instability of CKF-algo in real implementation so that the generated averaged atomic time is close to the theoretical value (or, equivalently, the value of JST-algo).
However, this method can not completely avoid numerical instability (see the different overlapping Allan deviation of the averaged atomic time of CKF-algo \cite{greenhall2012review,greenhall2006kalman} and JST-algo represented by the yellow and black lines in Fig.~\ref{Allan_1}).
In terms of the short time performance in 5 minutes, it can be seen from Fig.~\ref{Allan_1} that the overlapping Allan deviation of CKF-algo coincides with JST-algo (see the dashed line and the black markers). 
This is because the calculation errors are negligible at the beginning of calculations and hence the averaged atomic times are almost the same under CKF-algo and JST-algo in real implementation.

\begin{figure}
\centering
\includegraphics[width=8.5cm]{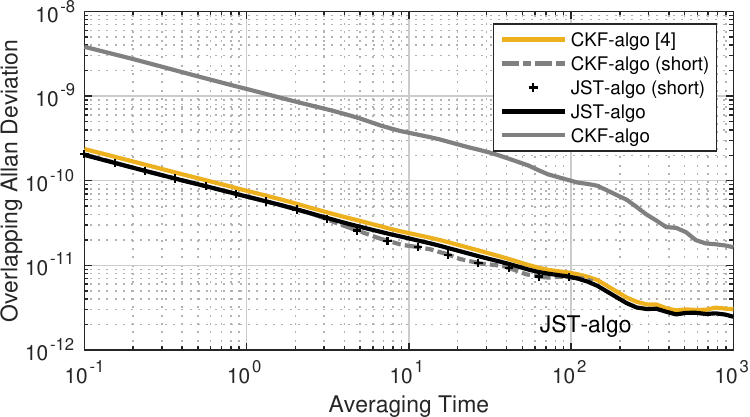} 
\caption{Overlapping Allan deviations of the time scale of  JST-algo, CKF-algo, and CKF-algo \cite{greenhall2012review} with equal weights.}\label{Allan_1}\vspace{-6pt}
\end{figure}

\begin{figure}
 \centering
\includegraphics[width=8.4cm]{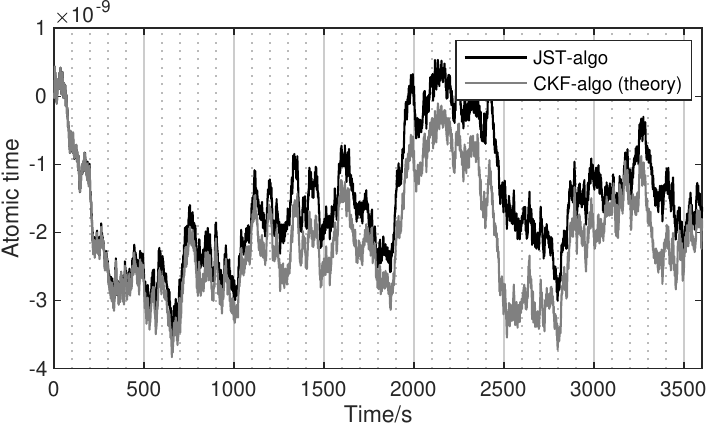} 
\caption{
The averaged atomic time ${\rm TA}[k]$ under JST-algo and CKF-algo (theory) with non-equal weights
}
\label{TA_2}\vspace{-5pt}
\end{figure}

\subsubsection{Nonequal Wights} Now, we consider the case with non-equal weights.
Let $ \beta=( 0.250, 0.375,0.125,0.125,0.1250)^\tr$, it follows from Theorem~\ref{prop:partial_kalman_obse} that since the necessary condition $\beta=\tfrac{1}{m} \mathds{1}_m$ for equivalence is not satisfied, CKF-algo and JST-algo can not generate the same averaged atomic time ${\rm TA}[k]$. 
This fact can be verified by the averaged atomic time shown in Fig.~\ref{TA_2}, where the black and grey lines correspond to ${\rm TA}[k]$ of JST-algo and CKF-algo in theory, respectively. 

\subsubsection{Runtime} Finally, it is interesting to note that JST-algo is superior to CKF-algo in the runtime when there are a number of atomic clocks in the ensemble. 
Figure~\ref{runtime} shows the runtime of JST-algo and CKF-algo versus the number $m$ of the clocks. 
It can be seen from the figure that the runtime of CKF-algo is likely to be exponentially increased when we increase the size of the ensemble, but the runtime of JST-algo is almost the same when increasing the number of clocks from 2 to 20.

\begin{figure}
 \centering
\includegraphics[width=7.9cm]{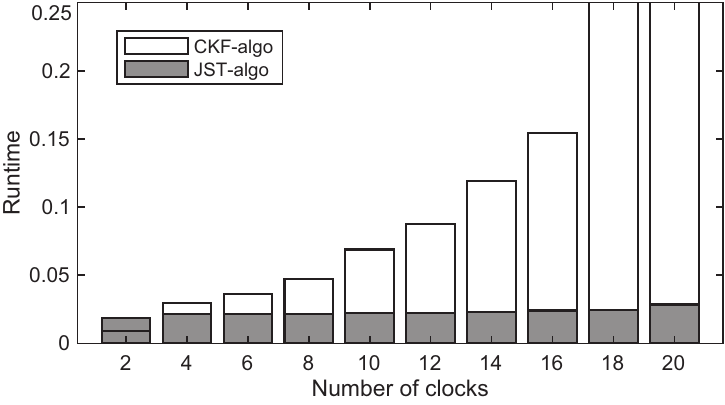} 
\caption{The runtime of CKF-algo and JST-algo versus the number $m$ of the clocks, where each runtime is taken as the averaged value of the runtimes of 500 simulations.}
\label{runtime}\vspace{-6pt}
\end{figure}

\begin{figure}
\centering
		\includegraphics[width=8cm]{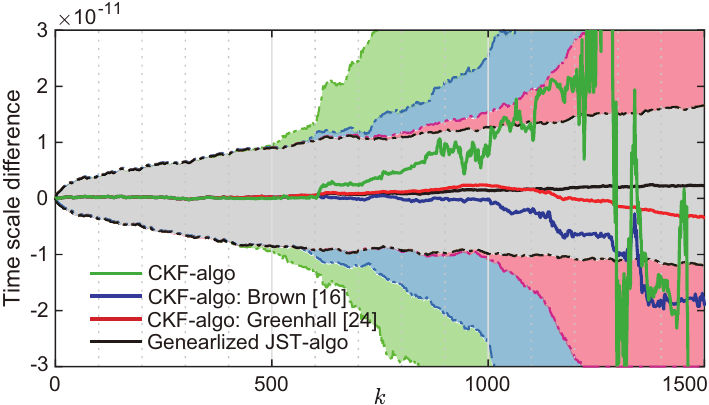} 
\caption{98$\%$-confidence interval of ${\rm TA}[k]$ under the CKF-algo, CKF-algo with Brown and Greenhall's correction, and the generalized JST-algo.
The solid line represents the mean of ${\rm TA}[k]$ in 50 times of simulations.}\label{TII_example2_confidence}\vspace{-6pt}
\end{figure}

\vspace{-8pt}
\subsection{Example 2: Third-order Clocks} 
Consider a third-order homogeneous atomic clock ensemble with $m=3$ clocks where variances of the system noises are set to  
 $\sigma_1^j=9e-26$,  $\sigma_2^j=7.5e-34$, and $\sigma_3^j=1e-47$ for all the clocks.  
The sampling period is set to $\tau=1$s.
In the simulation, both of the initial state $\bm{x}[0]$ and the guess of the initial state $\hat{\bm{x}}[0]$ of this 3-clock ensemble are set to   $\hat{\bm{x}}[0]={\bm{x}}[0]=1e-28\mathds{1}_9$. 
The guess of the  state noise covariance is set to the same as the actual one satisfying $\hat W=Q\otimes I_m$
for some $Q\geq0$.
Furthermore, we let  $\beta=\tfrac{1}{3} \mathds{1}_3$.

\begin{figure}
 \centering
\includegraphics[width=8.3cm]{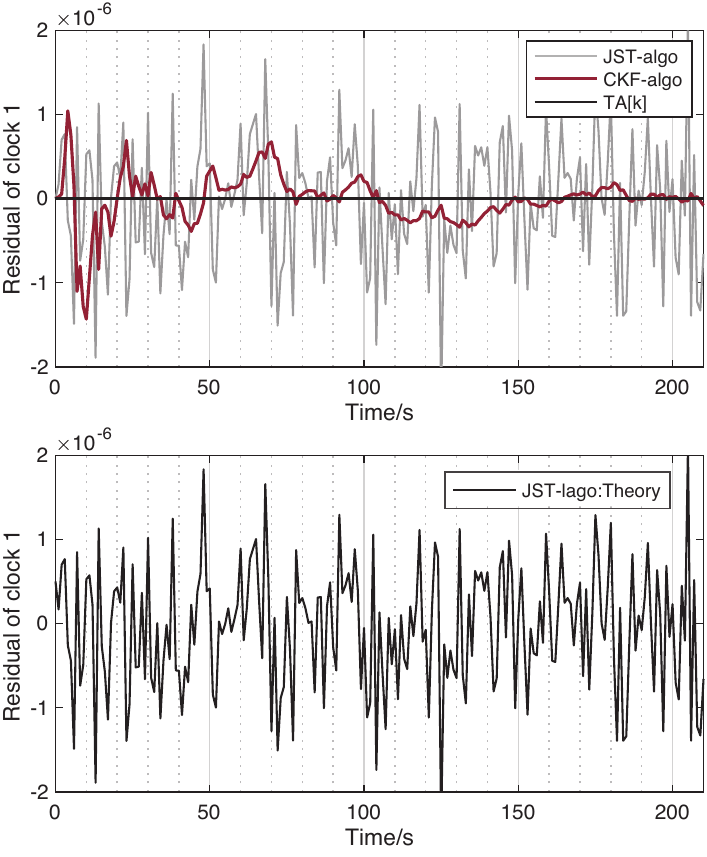} 
\caption{The actual and theoretical residual of clock 1 under the CKF-algo and the generalized JST-algo with large observation noises. The theoretical value for JST-algo is calculated by  \eqref{eq:weqs} in Theorem~\ref{prop:thm3}.}
\label{residual}\vspace{-10pt}
\end{figure}

\subsubsection{Confidence Interval of Averaged Atomic Time} 
The confidence interval of the averaged atomic time ${\rm TA}[k]$ 
 is shown in Fig.~\ref{TII_example2_confidence} with the variance of observation noises being $1e-12$ for all the clocks (i.e., $R=1e-12I_{3}$) to compare the generalized JST-algo, CKF-algo, and CKF-algo with Brown and Greenhall's correction. It can be seen from the figure that even though the confidence intervals of those algorithms coincide with each other in the early stage but they are diverse from each other in the later stage. 
 The generalized JST-algo can further narrow the confidence interval from CKF-algo with Brown's correction \cite{brown1991theory} and  Greenhall's correction   \cite{greenhall2006kalman}, meaning that numerical stability is improved.
 
\subsubsection{Clock Residual With Large Observation Noise} 
Note that $R=1e-12I_{3}$ is too large to satisfy the condition \eqref{eq:condition_R}.
It can be calculated that $\mathcal L_1=\mathcal L_2= 5.56e-13$, $\mathcal L_3= 2.22e-13$. 
Thus, since the inequality \eqref{eq:indival_better_cond} holds with the opposite signs,  
it follows from Theorem~\ref{prop:thm4} that CKF-algo is better than the generalized JST-algo in generating smaller variances of residual for all the 3 clocks. 
This result can be verified by the residual of clock 1 illustrated in Fig.~\ref{residual} where the red (resp., grey) line represents the one under CKF-algo (resp., JST-algo) with  $\bm P_0=1e-13I$.
It can be seen from this figure that the residual of clock 1 under the generalized JST-algo is exactly the same as the theoretical value calculated by  \eqref{eq:weqs}, which verifies Theorem~\ref{prop:thm3}.

\begin{figure}
 \centering
\includegraphics[width=8.3cm]{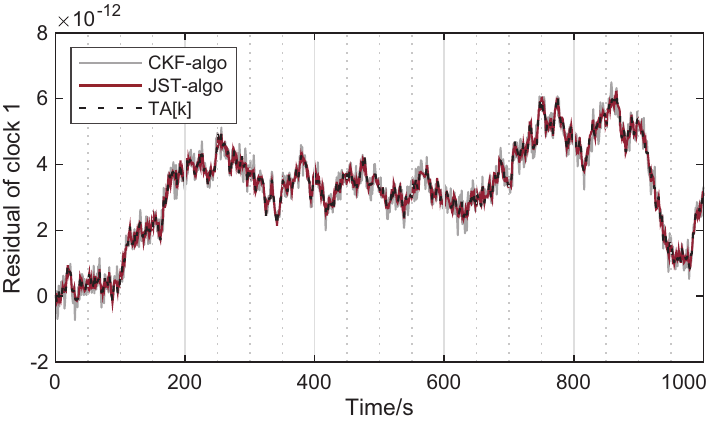} 
\caption{The actual residual of clock 1 under the CKF-algo and the generalized JST-algo with tiny observation noises. }
\label{regions2}\vspace{-8pt}
\end{figure}

\subsubsection{Clock Residual With Small Observation Noise}
Let the variance of observation noises be set to $1e-27$ for all the clocks, i.e., $R=1e-27I_{3}$, so that the conditions \eqref{eq:condition_R}, \eqref{eq:indival_better_cond} are satisfied with $\mathcal L_1=\mathcal L_2=-6.0000e-26$, and $\mathcal L_3=-6.0005e-26$.
It follows from Theorem~\ref{prop:thm4} that the generalized JST-algo is better than CKF-algo in generating smaller variances of individual residuals for all the 3 clocks.  
Without loss of generality, the result for clock 1 can be verified by the grey and red lines representing the residual of clock 1 of CKF-algo and the generalized JST-algo in Fig.~\ref{regions2}. 
In addition, we note that the red line in Fig.~\ref{regions3} represents theoretical difference $\mathcal L_1=\lim_{k\to \infty}
\{\mathbb C[\epsilon_1^{\rm JST}[k]]-\mathbb C[\epsilon_1^{\rm CKF}[k]] 
\}$ between the generalized JST-algo and CKF-algo in residual of clock 1, 
which is close to the actual value with many stochastic paths and hence verifies the results \eqref{eq:condition_R_skech} and \eqref{eq:indival_better}.

\begin{figure}
 \centering
\includegraphics[width=8.3cm]{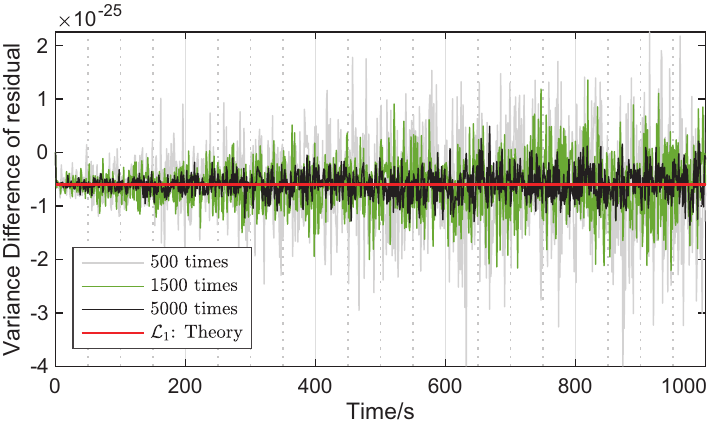} 
\caption{The actual and theoretical difference $\mathcal L_1$ between JST-algo and CKF-algo in the residual of clock 1. }
\label{regions3}\vspace{-6pt}
\end{figure}

\begin{figure}
 \centering
\includegraphics[width=8.3cm]{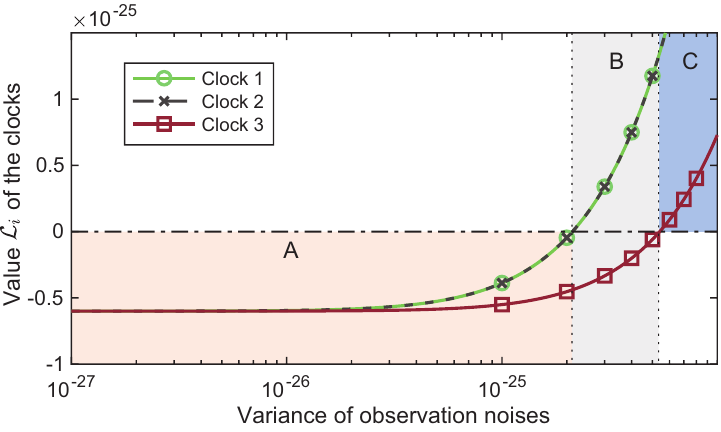} 
\caption{The value of $\mathcal L_i$ for clock $i=1,2$, and $3$ versus the variance $R=rI_{3}$ of observation noises. In the region A (resp., C), $\mathcal L_i<0$ (resp., $>0$) for all the clocks and hence the generalized JST-algo is better (resp., worse)  than CKF-algo in the individual residuals. In region B, the facts $\mathcal L_i>0$,  $i=1,2$, and $\mathcal L_3<0$ indicate that the generalized JST-algo is worse than CKF-algo in the individual residuals for clocks 1 and 2 but better for clock 3. }
\label{regions}\vspace{-6pt}
\end{figure}

\subsubsection{Discussion on Size of Observation Noise}\label{sec:Discussion}
Now we briefly discuss the relation between the variance $R=rI_{3}$ of observation noises and the superiority of the generalized JST-algo compared to CKF-algo.  
It can be seen from Fig.~\ref{regions} that when the observation noises are small enough (e.g., $r$ is in the region $A$ of the figure), $\mathcal L_i<0$ holds for all the clocks and hence the generalized JST-algo is better than CKF-algo in each of the individual residuals (even though the averaged atomic times ${\rm TA}[k]$ of the two methods are identical to each other as we discussed in Theorem~\ref{prop:partial_kalman_obse}). 
Alternatively, if the observation noises are large enough (e.g., $r$ is in the region $B$ of the figure), $\mathcal L_i>0$ holds for all the clocks and hence the generalized JST-algo is worse than CKF-algo in each of the individual residuals. 
Therefore, recalling the discussion in Section~\ref{sec:Discussion} about the runtime of JST-algo and CKF-algo, when the observation noises are small enough in such a homogeneous clock ensemble, it is suggested to imply the generalized JST-algo (instead of CKF-algo) for time scale generation especially in the case when the number of the clocks is large so that both the individual residuals and the runtime of the algorithms can be reduced.

\section{Conclusion} \label{sec:conc}
In this paper, we studied the comparison between two time scale generation algorithms using atomic clock ensembles, that are, CKF-algo and JST-algo.
We presented a generalized JST-algo via the state-space model of the higher-order atomic clock ensemble, where the proposed generalized JST-algo is reduced to the existing JST-algo for second-order clocks. 
By revealing the theoretical expressions of the averaged atomic times, we discussed the relation between the generalized JST-algo and CKF-algo. 
It is found that even though the measurement signal is not filtered,  
JST-algo can yield the averaged atomic times independent to the observation noise. 
The prediction error of Kalman filtering algorithm was rigorously shown by using the prediction error regarding an observble state space.
We proved that the generalized JST-algo is equivalent to CKF-algo in the sense of generating averaged atomic time if and only if equal averaging weights are considered for the atomic clocks when the covariance matrices of system noises are identical for all the clocks.

In such a homogeneous clock ensemble, we further revealed the theoretical relation between the generalized JST-algo and CKF-algo in the individual clock residuals and presented the sufficient and necessary condition for observation noises to determine which algorithm can generate the clock residuals with smaller variances. 
We discussed the relation between the runtime of the algorithms versus the number of clocks in one of the numerical examples. 
It is found that if the observation noises are tiny enough by some sophisticated measuring equipment, one is suggested to imply the generalized JST-algo (instead of CKF-algo) for time scale generation since the generalized JST-algo can generate smaller clock residuals and the calculation cost is lower than CKF-algo if the number of the atomic clocks is large.

\vspace{-6pt}

\bibliographystyle{IEEEtran} 
 
\bibliography{Kalman}

\appendix
\subsubsection*{Deviation of \eqref{eq:PKA2_o}} 
Note that
\begin{equation}\nonumber
( I_n \otimes \overline{V})  \bm P_k \bm H^{\sf T} 
%= ( I_n \otimes \overline{V}) \bm P_k ( I_n \otimes \overline{V})^{\sf T}\bm H_{\rm o}^{\sf T}
= \hat{\bm P}_k  ^{\tr}  \bm H_{\rm o}^{\sf T}  , \ \ 
 \bm H    \bm P_k \bm H^{\sf T} 
%= \bm H_{\rm o}  ( I_n \otimes \overline{V}) \bm P_k( I_n \otimes \overline{V})^{\tr}\bm H_{\rm o}^{\sf T} 
=  \bm H_{\rm o}  \hat{\bm P}_k ^{\sf T}  \bm H_{\rm o}^{\sf T}.
\end{equation}
imply that the Kalman gain in observable state is given by
\begin{align} \nonumber
 \hat{\bm K}_k=& ( I_n \otimes \overline{V})\bm K_k  = ( I_n \otimes \overline{V})     \bm P_k {{\bm H} ^{\sf T}\bigr(\bm{H}\bm P_k{ \bm H}^{\sf T}+R\bigr)^{-1}} \\  
&\qquad\quad\  \quad\quad=  \hat{\bm P}_k \bm H_{\rm o}^{\sf T} (\bm H_{\rm o} \hat{\bm P}_k \bm H_{\rm o}^{\sf T}+R)^{-1}
\end{align}
whereas the covariance is given by 
\begin{align} \nonumber
 \hat{\bm P}_{k+1}    =&
( I_n \otimes \overline{V})\bm F[k] \bm P_k\bm F[k]^{\sf T} ( I_n \otimes \overline{V}^{\sf T})+\hat W_{\rm o}    \\ \nonumber
& - (I_n \otimes \overline{V}) \bm F[k] \bm K_k {\bm H}\bm P_k \bm F[k]^{\sf T}( I_n \otimes \overline{V}^{\sf T})  \\ \nonumber
= & \bm F_{ \rm o}[k]( I_n \otimes \overline{V}) \bm P_k ( I_n \otimes \overline{V}^{\sf T})\bm F_{ \rm o}[k]^{\sf T}+\hat W_{\rm o} \\ \nonumber
& -\bm F_{ \rm o}[k]( I_n \otimes \overline{V})\bm K_k \bm H \bm P_k ( I_n \otimes \overline{V}^{\sf T})\bm F_{ \rm o}[k]^{\sf T}\\  
= &\bm F_{ \rm o}[k](  \hat{\bm P}_{k}  -\bm N_k    \bm H_{\rm o}  \hat{\bm P}_{k} ) \bm F_{ \rm o}[k]^{\sf T}  +\hat W_{\rm o} .
\end{align}
Thus, the equation \eqref{eq:PKA2_o} is obtained since we have
\begin{align} \nonumber
 \hat{\bm \xi}_{\rm o} [k+1]    
 =& ( I_n \otimes \overline{V}) \bm F[k] \hat{\bm x}[k] + ( I_n \otimes \overline{V})\bm K_k(\bm{y}[k]-\bm{H}\hat{\bm x}[k]) \\  
  =&  \bm F_{ \rm o}[k]  \hat{\bm \xi}_{\rm o} [k] +\hat{\bm K}_k ( \bm y[k]-\bm{H}_{\rm o} \hat{\bm{\xi}}_{\rm o}[k]).
\end{align}

\vspace{-10pt}
 \begin{IEEEbiography}{Yuyue Yan}
(S'19--M'22) received the B.Eng. degree in electronic information engineering from Xiamen University Tan Kah Kee College,  Fujian, China, in 2017, and the M.E. degree and the Ph.D. degree in systems and control engineering from Tokyo Institute of Technology, Tokyo, Japan, in 2019 and 2022, respectively.

He is currently a Research Fellow with the Department of Systems and Control Engineering, Tokyo Institute of Technology. His research interests include stability of noncooperative dynamical system, transportation system, and distributed time synchronization with atomic clock ensembles.
\end{IEEEbiography}
\vspace{-10pt}
\begin{IEEEbiography}{Takahiro Kawaguchi}(Member, IEEE)
received the B.Sc., M.Sc., and Ph.D. degrees in engineering from Keio University, Tokyo, Japan, in
2011, 2013, and 2017, respectively.
From 2013 to 2015, he was with the Toshiba Research and Development Center. From 2017 to 2019, he was a Researcher with the Department of Systems and Control Engineering, School of Engineering, Tokyo Institute of Technology, Tokyo. 
From 2019 to 2020, he was a specially appointed Assistant Professor with the Department of Systems and Control Engineering, Tokyo Institute of Technology. He is currently an Assistant Professor with the Division of Electronics and Informatics, Graduate School of Science and Technology, Gunma
University, Gunma, Japan. 
His research interests include system identification
theory and the application of machine learning techniques.
Dr. Kawaguchi is a member of the Society of Instrument and Control Engineers and the Institute of System, Control, and Information Engineers.
\end{IEEEbiography}
\vspace{-10pt}
\begin{IEEEbiography}{Yuichiro Yano}
(Member, IEEE)
received Ph.D. in engineering from Tokyo Metropolitan University in 2015. 
From April 2014 to March 2016, he was research fellowship for young scientists at Japan Society for the Promotion of Science (JSPS). 
From April 2016, he has worked as tenure-track researcher with National Institute of Information and Communications Technology (NICT), Tokyo, Japan. Since April 2019, he has been a permanent researcher with same institute.
\end{IEEEbiography}

\begin{IEEEbiography}{Yuko Hanado}
received BS and MS degrees in Science from Tohoku University in 1987 and 1989 respectively, and Ph.D. degrees in Graduated school of Information Systems from the University of Electro-Communication in 2008. 
In 1989 she joined NICT and is currently the Director General of Electromagnetic Standards Research Center in Radio Research Institute at NICT. 
She has been engaged in the work for time and frequency standards, and especially has interests to algorithm of making ensemble atomic timescale. 
She was the recipient of Award of the Commendation for Science and Technology by The Minister of Education, Culture, Sports, Science and Technology (MEXT) in 2013.
\end{IEEEbiography}

 \begin{IEEEbiography}{Takayuki Ishizaki} 
(Member, IEEE)
received the B.Sc., M.Sc., and Ph.D. degrees in Engineering from Tokyo Institute of Technology, Tokyo, Japan, in 2008, 2009, and 2012, respectively.
Since November 2012, he has been with Tokyo Institute of Technology, where he is currently an Associate Professor at the Department of Systems and Control Engineering. 
His research interests include network systems control,  power systems applications, and distributed time synchronization with atomic clock ensembles.

He was the recipient of awards including Pioneer Award of Control Division from The Society of Instrument and Control Engineers (SICE) in 2019, IEEE Control Systems Magazine Outstanding Paper Award from IEEE Control Systems Society (IEEE CSS) in 2020, and The Young Scientists' Award of the Commendation for Science and Technology by The Minister of Education, Culture, Sports, Science and Technology (MEXT) in 2021.
\end{IEEEbiography}

\end{document}